\documentclass{article}
\usepackage[utf8]{inputenc}

\usepackage{graphicx}
\usepackage{amssymb,amsthm}
\usepackage{enumerate}
\usepackage[T1]{fontenc}

\newcommand{\eps}{\varepsilon}
\renewcommand{\sc}{\{ c \}}

\newcommand{\RG}{\ensuremath{G^{Red}}}

\newcommand{\RV}{\ensuremath{V^{Red}}}
\newcommand{\RE}{\ensuremath{E^{Red}}}

\newcommand{\FFI}{\textsc{2-Free-Flood-It }}

\newcommand{\ie}{\emph{i.e. \ }}

\newtheorem{thm}{Theorem}
\newtheorem{lemme}{Lemma}

\newtheorem{cor}{Corollary}

\title{\FFI is polynomial}

\author{Aurélie Lagoutte\\ LIP, ENS Lyon}

\date{\today}

\begin{document}

\maketitle

\begin{abstract}
We study a discrete diffusion process introduced in some combinatorial games called FLOODIT and MADVIRUS that can be played online~\cite{jeuFI,jeuMV} and whose computational complexity has been recently studied~\cite{Arthur,Anglais}. 

The flooding dynamics used in those games can be defined for any colored graph. It is shown in~\cite{Aurelie-hal} that studying this dynamics directly on general graph is a valuable approach to understand its specificities and extract uncluttered key patterns or algorithms that can be applied with success to particular cases like the square grid of FLOODIT or the hexagonal grid of MADVIRUS, and many other classes of graphs.   

This report is the translation from french to english of the section in~\cite{Aurelie-hal} showing that the variant of the problem called \FFI can be solved with a polynomial algorithm, answering a question raised in~\cite{Arthur,Anglais}.
\end{abstract}

\section{Definitions and notation} 
\label{intro}

Let $G=(V,E)$ be a connected undirected graph, with vertices $V$ and edges $E$. The number of vertices (resp. edges) will be denoted $n$ (resp. $m$). A {\em coloration} of~$G$ is a mapping from~$V$ into a set of colors~$C$. It will be called a {\em $c$-coloration} if $|C|=c$. It will be called a {\em proper coloration} if adjacent vertices have different colors. Once a coloration of~$G$ is given, a {\em zone} $Z$ is defined as a connected monochromatic subset of~$V$. The dynamics which is studied consists in applying a sequence of {\em flooding operations} to an initial colored graph. In the FREE-FLOOD-IT version~\cite{Arthur,Anglais} that is studied here, a {\em flooding operation} consists in choosing a zone $Z$ and a color~$c$ and then replacing the color of all vertices in~$Z$ by~$c$. It yields a new coloration of the initial graph where the zone~$Z$ may extend if some adjacent zones were colored by~$c$. The game associated to this dynamics takes a colored graph as input and aims at finding the shortest length of a sequence leading to a monochromatic graph. At that time, we say that we have {\em flooded the whole graph}. Finding this shortest length is an optimization problem, which is hard for general graphs when working with 3 or more colors (NP-hardness results in~\cite{Arthur,Anglais}). But for 2 colors, the optimization problem that will be called \FFI has a computational complexity which remained open (the question was raised in~\cite{Arthur,Anglais} for the class of square grids $N \times N$). We show in this report that there exists a polynomial algorithm to solve this problem for general graphs (and thus for square grids). We need some further definitions and notation.

Let $G=(V,E)$ be a connected undirected graph and a $c$-coloration of this graph. We define its {\em reduced graph}, denoted $\RG=(\RV,\RE)$, as the connected undirected graph where $\RV$ is the set of all the zones of~$G$ and $\RE$ puts an edge between two zones of~$G$ if there exists an edge with extremities in both zones. Note that this reduced graph is clearly smaller than $G$: it is a minor of $G$, and $|\RV| \leq [V|$ and $|\RE| \leq |E|$. By associating to each zone its color in~$G$, one gets a $c$-coloration of the reduced graph and this is a proper coloration. It should be clear that for the flooding process, working with the initial graph or its reduced version is {\em perfectly equivalent}. Moreover it is easy to find a linear algorithm that computes the reduced graph of a colored graph~\cite{Aurelie-hal}.    

From now on, we will focus on \FFI where we only consider 2-colorations. We also decide to work exclusively with the reduced graphs (note that with proper 2-colorations, those reduced graphs are bipartite). At each flooding step, we start from a reduced graph $G=(V,E)$ and its proper 2-coloration. Then the flooding operation just consists in choosing a vertex $x$ (since we only have 2 colors, the choice of the flooding color is imposed if we want to modify the coloration). We perform the flooding: here the zones are reduced to the single vertices and we just change the color of vertex~$x$. Since we start from a proper 2-coloration, the vertex~$x$ has now the same color as all its neighbors $N(x)=\{y \in V | xy \in E\}$. We can compute the new reduced graph: it is obtained by contracting $x$ and its neighbors into one single vertex (that we still denote~$x$). More precisely, if we start from the graph $G$ (as a matter of fact in our context $G$ is bipartite and there is no need to specify the coloration since there are only two identical proper 2-colorations obtained by switching colors) and we decide to perform a flooding operation at vertex~$x$, then the new reduced graph obtained after the flooding operation is $G/x=(V/x, E/x)$ where $V/x=V \setminus N(x)$ and  $E/x=(E \setminus \{xy \in E | y \in N(x) \}) \cup \{xz \ | \ \exists y \in N(x) \quad yz \in E \}$. We call this transformation of~$G$ into~$G/x$ a {\em neighborhood contraction}, it is a particular case of edge contractions well-known in the definition of graph minors. For any path $\gamma$ in~$G$, we will denote $\gamma /x$ the path in~$G/x$ obtained after the contraction around~$x$, that is by replacing any occurrence of neighbors of~$x$ by the vertex~$x$ itself. Note that any path in~$G/x$ can be described as $\gamma/x$ for at least one path $\gamma$ in~$G$.   

Finally we give back a few classical definitions about distances in undirected graphs. Let $G=(V,E)$ be a connected undirected graph, we denote $d(x,y)$ the classical distance between vertices~$x$ and~$y$. For any vertex $x \in V$, the {\em eccentricity} of~$x$ is $r(x)=\max \{d(x,y) \ | y\in V\}$. The {\em radius} of~$G$ is  $R(G)=\min \{r(x) \ | x \in V\}$. The {\em center} of~$G$ is $C(G)=\{ x \in V \ | r(x)=R(G)\}$. It is well-known that for any connected graph given by its adjacency lists, computing the eccentricity of a vertex~$x$ can be done in~$O(m)$ time by performing a BFS from~$x$. Consequently its radius and its center can be computed in~$O(nm)$ time.  

\section{A simple formula for \FFI}

With this approach of systematically reducing the graph after each flooding, solving a \FFI instance, \ie finding a shortest sequence of flooding operations which floods the whole graph, comes to finding a shortest sequence of vertices such that the corresponding neighborhood contractions reduce the graph to a single vertex. We are going to show that this minimum number of steps is exactly the radius of the initial reduced graph. But first it requires a few results to study how the radius behaves with regard to neighborhood contractions. Note that in the case of \FFI the initial reduced graph is bipartite, but all our next lemmas and theorems apply to arbitrary graphs. 

\begin{lemme}\label{autre sommet loin}
Let $G=(V,E)$ be a connected graph. Let $c \in C(G)$, $y\in V$ such that $d(c,y)=R(G)$ and let $\gamma=c a_1\ldots a_{r-1}y$ be a shortest path from $c$ to $y$ (\ie $r=R(G)$). Then\\
\begin{enumerate}[(i).]
\item \label{prop pas assez precise} There exists $z \neq y$ such that for any shortest path $\mu$ from $c$ to $z$, $\gamma \cap \mu=\sc$ and $R(G)-1\leq d(c,z)\leq R(G)$.
\item More precisely: either there exists $z \in V$ satisfying all conditions of (\ref{prop pas assez precise}) and $d(c,z)=R(G)$, or for all $z \in V$ satisfying all conditions of (\ref{prop pas assez precise}), we have $d(c,z)=R(G)-1$ and there exists some vertex $z_0$ among those vertices such that for any path $\mu$ from $c$ to $z_0$, of length $d(c,z_0)+1$, we have $\gamma \cap \mu = \sc$.
\end{enumerate}
\end{lemme}

\begin{proof}\label{preuve_lemme_autre_sommet_loin}
\begin{enumerate}[(i).]

\item Let us first show that there exists some $z \in V$ different from $y$ such that $R(G)-1\leq d(c,z)\leq R(G)$.

Assume that it is not true, then for all $z \in V$ different from $y$, we have $d(c,z)\leq R(G)-2$. Then $d(a_1,y)=R(G)-1$. Moreover, for all $z\in V$ different from $y$, $d(a_1,z)\leq d(a_1,c)+d(c,z)\leq 1+R(G)-2 = R(G)-1$. Thus, $r(a_1)=R(G)-1$. It yields a contradiction with $R(G)=\min \{r(x)| x \in V\}$.

Now let us show property (\ref{prop pas assez precise}). Once again assume that it is not true. We have shown just before that there exists some $z \in V$ satisfying $R(G)-1\leq d(c,z)\leq R(G)$. Let us denote $z_1,\ldots,z_n$ all the vertices different from $y$ satisfying the inequalities. Then for all $i\in  \{1,\ldots,n\}$, there exists some shortest path $\mu_i=c b_{i,1} \ldots b_{i,r-\eps_i} z_i$ (with $\eps_i \in \{1;2\}$) from $c$ to $z_i$ such that $\gamma \cap \mu_i \neq \sc$. It implies that for all $i\in  \{1,\ldots,n\}$ there exist $k_i,\ j_i \in \{1,\ldots,r\}$ such that $a_{k_i}=b_{j_i}$. Since $\gamma$ and $\mu$ are shortest paths, we have $k_i=j_i$. Consequently 
\begin{center}
$\begin{array}{lr@{\ \leq \ }l}
 \forall i\in  \{1,\ldots,n\} \qquad & d(a_1,z_i) & |a_1 \ldots a_{k_i} b_{i,k_i+1} \ldots
b_{i,r-\eps_i} z_i| \\
&& r-\eps_i \\
&& R(G)-1\\
\end{array}$
\end{center}
Moreover for any vertex $x \in V \setminus \{z_1,\ldots,z_n,y\}$,
\begin{center}
 $\begin{array}{r@{\ \leq \ }l}
  d(a_1,x) & d(a_1,c)+d(c,x)\\
 & 1+R(G)-2\\
 & R(G)-1\\
 \end{array}$
\end{center}
We also have $d(a_1,y)=|a_1 \ldots a_{r-1} y|=r-1=R(G)-1$.

Finally it implies $r(a_1)=R(G)-1$, which yields a contradiction with $R(G)=\min \{r(x)| x \in V\}$.

\item Let $z_1,\ldots, z_n$ be all the vertices satisfying all the conditions of (\ref{prop pas assez precise}). Assume that for all $i$, $d(c,z_i)=R(G)-1$. Now we have to show that there exists some $i\in \{1, \ldots n\}$, such that any path $\mu_i$ from $c$ to $z_i$ of length $d(c,z_i)+1$ satisfies $\mu_i \cap \gamma=\sc$. 

Assume that it is not true, then for all $i$, there exists a path $\mu_i$ from $c$ to $z_i$, of length $d(c,z_i)+1$ such that $\mu_i \cap \gamma \neq \sc$. Let us denote $\mu_i=c b_{i,1} \ldots b_{i,r-1} z_i$. Then for all $i$, there exists some $k_i$ such that $a_{k_i}=b_{i,k_i}$ or $a_{k_i}=b_{i,k_i+1}$. Then
\begin{itemize}
 \item $d(a_1,y)=R(G)-1$
 \item for all $i$, we have $d(a_1,z_i)= |a_1 \ldots a_{k_i} b_{i,k_i+\eps_i} \ldots b_{i,r-1} z_i|$ where $\eps_i \in \{1;2\}$ thus $d(a_1,z_i) \leq R(G)-1$.
 \item for all $x$ different from the vertices $z_i$ satisfying the inequalities $R(G)-1\leq d(c,x) \leq R(G)$, there exists a shortest path $\delta=c d_1 \ldots d_{r-\eps} x$ (with $\eps \in \{1,2\}$) from $c$ to $x$ such that $\delta \cap \gamma \neq \sc$. Then there exist $k_i,\ j_i$ such that $a_{k_i}=b_{j_i}$. Since both $\gamma$ and $\delta$ are shortest path, we have $k_i=j_i$. Thus we have $d(a_1,x) \leq |a_1 \ldots a_{k_i} b_{i,k_i+1} \ldots b_{i,r-\eps} z_i|$ which implies $d(a_1,x) \leq r-\eps \leq R(G)-1$.
 \item for all $x$ such that $d(c,x)\leq R(G)-2$ (it covers all the other cases), we have: $d(a_1,x) \leq d(a_1,c) +d(c,x) \leq R(G)-1$.
\end{itemize}
Finally $r(a_1)=R(G)-1<R(G)$, which yields a contradiction once again.
\end{enumerate}
\end{proof}

\begin{lemme}\label{longueur chemin}
Let $a,b,x$ be three vertices of $G$ a connected graph.
\begin{enumerate}[(i).]
\item Assume that there exists no shortest path from $a$ to $b$ using vertex $x$. Then
\begin{center}
 $d_{G/x}(a,b)\geq d_G(a,b)-1$
\end{center}
and the equality is achieved if and only if there exists a path of length $d_G(a,b)+1$ from $a$ to $b$ using $x$ in $G$.
\item Otherwise, we have
\begin{center}
  $d_{G/x}(a,b)\geq d_G(a,b)-2$
\end{center}
\end{enumerate}
\end{lemme}

\begin{proof}
\begin{enumerate}[(i).]
\item Suppose it is not true, then there exists a shortest path $ \lambda $ from $a$ to $b$ in $G/x$, of length $\leq d_G(a,b)-2 $. Let $\mu$ be a path in $G$ such that $\mu/x=\lambda$. This path $\mu$ goes from $a$ to $b$, which ensures that $|\mu|\geq d_G(a,b)\geq |\lambda|+2$. Such a reduction of the length requires that $\lambda$ goes through $x$ in $G/x$. Since $\lambda$ is a shortest path, it can be written as $\lambda=\lambda_1 x
\lambda_2$ where $\lambda_1$ and $\lambda_2$ are both shortest path not going through $x$ in $G/x$. Thus, $\lambda_1$ and $\lambda_2$ have not been changed by the neighborhood contraction around $x$ and they are also shortest path in $G$. As a consequence, $\mu=\lambda_1 x_1 \ldots x_n \lambda_2$ où $\{x_1,\ldots,x_n\} \subseteq N(x)\cup \{x\}$ and $n=|\mu|-|\lambda|+1$. In this way, one can build in $G$ a path $\mu'=\lambda_1 x_1 x x_n \lambda_2$ of length $|\lambda|+2 \leq d_G(a,b)$ going through $x$. Note also that $\mu'$ goes from $a$ to $b$ in $G$, thus $|\mu'|\geq d_G(a,b)$. Finally, $|\mu'|=d_G(a,b)$ thus $\mu'$ is a shortest path from $a$ to $b$ in $G$, going through $x$: it contradicts the initial assumption about $a$,$b$,$x$.

\emph{Case of equality:}
\begin{itemize}
 \item $\Rightarrow$ There exists a shortest path $\lambda$ from $a$ to $b$ in $G/x$, of length $d_G(a,b)-1$. A construction of $\mu'$ similar to the one just above provides a path of length $d_G(a,b)+1$ going from $a$ to $b$ going through $x$ in $G$.
 \item $\Leftarrow$ Let $\mu'$ be a path of length $d_G(a,b)+1$ from $a$ to $b$ going through $x$ in $G$. The one can decompose $\mu$ into $\mu'=\mu_1 x_1 x x_2 \mu_2$, where $x_1$ and $x_2$ are neighbors of $x$. Then $|\mu'/x|\leq |\mu'|-2 = d_G(a,b)-1$.
\end{itemize}
 \item Let $\lambda$ be a shortest path from $a$ to $b$ in $G/x$. As above one can build a path $\mu'$ from $a$ to $b$ in $G$ such that $|\mu'|=|\lambda|+2$. Thus $d_G(a,b) \leq d_{G/x} (a,b)+2$, that is $d_G(a,b)-2 \leq d_{G/x} (a,b)$.
\end{enumerate}
\end{proof}

\begin{thm}\label{th contraction}
 Let $G=(V,E)$ be a connected graph and $x \in V$. Then $R(G)-1 \leq R(G/x) \leq R(G)$.
\end{thm}

\begin{proof}
It is obvious that for all $y,z \in V$, we have $d_{G/x}(y,z) \leq d_G(y,z)$ and thus $R(G/x) \leq R(G)$. Let us show that $R(G)-1 \leq R(G/x)$. It is clear that at least one of the following properties is satisfied:
\begin{enumerate}
  \item $x \in C(G)\cap C(G/x)$;
  \item there exists $c \in C(G)\cap C(G/x)$, $c\neq x$, and there exists $y \in V$ and $\gamma$ a shortest path from $c$ to $y$ such that $|\gamma|=R(G)$ and $x \in \gamma$;
  \item $C(G/x) \cap C(G)\neq \emptyset$ and for all $c \in C(G)\cap C(G/x)$ and for all $y \in V$ such that $d_G(c,y)=R(G)$, there is no shortest path from $c$ to $y$ going through $x$ in $G$;
  \item $C(G/x) \cap C(G)=\emptyset$.
\end{enumerate}
Let us prove the inequality by a case study:
\begin{enumerate}
  
  \item Since $x \in C(G)$, there exists $y \in V$ and $\gamma$ a shortest path from $x$ to $y$ such that $|\gamma|=R(G)$. Let us write $\gamma=x a_1 \ldots a_{r-1} y$, with $r=R(G)$. Since $\gamma$ is a shortest path, we can write $\gamma/x=x a_2 \ldots a_{r-1} y$ and we have $d_{G/x}(x,y)=|\gamma/x|$. It implies that $r_{G/x}(x) \geq d_{G/x}(x,y)=R(G)-1$ thus, since $x \in C(G/x)$, we get $R(G/x)=r_{G/x}(x) \geq R(G)-1$.
  
  \item From Lemma \ref{autre sommet loin}, one of the following cases occurs:
    \begin{itemize}
      \item There exists $z \in V$ such that for any shortest path $\mu$ from $c$ to $z$, $\gamma \cap \mu=\sc$ and $d_G(c,z)=R(G)$. Then $x$ does not belong to any shortest path from $c$ to $z$, thus from Lemma \ref{longueur chemin}.(i), $d_{G/x}(c,z) \geq d_G(c,z)-1 = R(G)-1$. Since $c \in C(G/x)$, we get $R(G/x)=r_{G/x}(c)\geq d_{G/x}(c,z)\geq R(G)-1$.
      \item There exists $z \in V$ such that for any shortest path $\mu$ from $c$ to $z$, $\gamma \cap \mu=\sc$ and $d_G(c,z)=R(G)-1$. Moreover, for any path $\mu'$ from $c$ to $z$ of length $|\mu'|=d_G(c,z)+1=R(G)$, we have $\gamma \cap \mu'=\sc$. From Lemma \ref{longueur chemin}.(i), $d_{G/x}(c,z)\geq R(G)-1$. Thus $R(G/x)=r_{G/x}(c)\geq d_{G/x}(c,z)\geq R(G)-1$.
    \end{itemize}
  
  \item Let $c \in C(G)\cap C(G/x)$ and $y \in V$ such that $d_G(c,y)=R(G)$ (the existence of such a vertex $y$ is ensured by the definition of $R(G)$). Then, from Lemma \ref{longueur chemin}.(i), we have the inequality $d_{G/x}(c,y) \geq d_G(c,y)-1 = R(G)-1$ which implies $R(G)=r_{G/x}(c) \geq d_{G/x}(c,y) \geq R(G)-1$.
  
  \item Let $c \in C(G/x)$. Then $r_G(c)\geq R(G)+1$. From Lemma \ref{longueur chemin}.(ii), we have:
     \begin{center}$
        \begin{array}{lr@{\ \leq \ }l}
        \forall y \in V, & d_G(c,y)-2 & d_{G/x}(c,y) \\
        \mathrm{donc} & \max \{d_G(c,y) | y \in V\} -2 &\max \{d_{G/x}(c,y) | y \in V/x \}\\
        \mathrm{donc} & R(G)+1-2 & R(G/x)\\
        \mathrm{donc} & R(G)-1 & R(G/x)\\
        \end{array}$
     \end{center}
             
\end{enumerate}
In all cases, we finally get $R(G/x) \geq R(G)-1$.
\end{proof}

\begin{lemme}\label{contraction depuis le centre}
If $c \in C(G)$, then $R(G/c)=R(G)-1$.
\end{lemme}

\begin{proof}
If $y \in V$, then $d_{G/c}(c,y)~\leq~d_G(c,y)-1$. As a matter of fact, if $\gamma=c~a_1~\ldots~a_{r-1}~y$ is a shortest path from $c$ to $y$, then $\gamma/c=c~a_2~\ldots~a_{r-1}~y$. Consequently $d_{G/c}(c,y)~\leq~|\gamma/c|~\leq~d_G(c,y)-1$. Looking at the maximum, we get the inequality: $\max~ \{d_{G/c}(c,y)|y~\in~V\}~\leq~\max\{d_G(c,y)|y \in V\}-1$, which is equivalent to $r_{G/c}(c)\leq~r_G(c)-1=R(G)-1$. It implies that $R(G/c)~\leq~R(G)-1$. From Theorem \ref{th contraction}, $R(G/c)~\geq R(G)-1$. It means that $R(G/c)= R(G)-1$.
\end{proof}

\begin{thm}\label{recurrence rayon}
Consider an instance of \FFI, that is an initial connected graph with a 2-coloration, and let $G$ be its reduced graph, that is a bipartite graph with a proper 2-coloration. Then the minimum number of steps required to flood the whole graph $G$ (or equivalently the whole initial graph) is exactly $R(G)$, the radius of the reduced graph $G$.  
\end{thm}

\begin{proof}
From Lemma \ref{contraction depuis le centre}, $R(G)$ is an upper bound of this number of steps: one can flood the whole graph by choosing a vertex $c$ in the center $C(G)$, and by repeatedly choosing this vertex to perform each flooding step. The initial graph is all flooded when one reaches a singleton graph, that is a graph of radius~0. Each step decreases the radius by exactly 1, it requires $R(G)$ steps. 

Now Theorem \ref{th contraction} shows whatever the choice of the vertex for a flooding, the radius decreases by at most 1. It implies that whatever the sequence of vertices chosen for the flooding operations, at least $R(G)$ steps are required to reach radius~0.
\end{proof}

\begin{cor}\label{corollaire 2FFI}
Consider an instance of \FFI, that is a connected graph with $n$ vertices and $m$ edges with a 2-coloration. Then a sequence of vertices of minimum length which floods the whole graph can be computed in time $O(nm)$. In the particular case studied in \cite{Arthur,Anglais} where the graph is a square grid $N \times N$, it can be computed in time $O(N^4)$.
\end{cor}

\begin{proof}
One first has to compute the reduced graph. It can be easily done in $O(m)$ time as mentioned before. The reduced graph has $\leq n$ vertices and $\leq m$ vertices. Then one only has to compute its radius in $O(nm)$ time. In the particular case of $N \times N$ square grid, we have $n=N \times N$ and $m=4 N^2-4N$.
\end{proof}

%\bibliographystyle{plain}
%\bibliography{ref_biblio_uk}

\end{document}